\documentclass{article}
\usepackage{graphicx}
\usepackage[latin9]{inputenc}
\usepackage{amsmath}
\usepackage{amsthm}
\usepackage{amssymb}
\usepackage{esint}
\usepackage{enumitem}
\usepackage[english]{babel}
\usepackage{latexsym}
\usepackage{cite}
\usepackage{dsfont}
\usepackage{babel}
\usepackage{color}
\usepackage[overload]{empheq}
\usepackage{algorithm}
\usepackage{algorithmicx}
\usepackage{algpseudocode}
\usepackage{nicefrac}

\inputencoding{utf8}

\makeatletter
\newenvironment{breakablealgorithm}
  {
   \begin{center}
     \refstepcounter{algorithm}
     \hrule height.8pt depth0pt \kern2pt
     \renewcommand{\caption}[2][\relax]{
       {\raggedright\textbf{\ALG@name~\thealgorithm} ##2\par}%
       \ifx\relax##1\relax 
         \addcontentsline{loa}{algorithm}{\protect\numberline{\thealgorithm}##2}%
       \else 
         \addcontentsline{loa}{algorithm}{\protect\numberline{\thealgorithm}##1}%
       \fi
       \kern2pt\hrule\kern2pt
     }
  }{
     \kern2pt\hrule\relax
   \end{center}
  }
\makeatother

\makeatletter
\renewcommand{\Function}[2]{%
  \csname ALG@cmd@\ALG@L @Function\endcsname{#1}{#2}%
  \def\jayden@currentfunction{#1}%
}
\newcommand{\funclabel}[1]{%
  \@bsphack
  \protected@write\@auxout{}{%
    \string\newlabel{#1}{{\jayden@currentfunction}{\thepage}}%
  }%
  \@esphack
}
\makeatother

\DeclareMathOperator{\avar}{\mathbb A\mathbb V@R}
\DeclareMathOperator*{\Pep}{{\mathcal P}_{\!\!\varepsilon}}

\DeclareMathOperator{\Ph}{\mathbb P_0}

\makeatletter
\@ifundefined{showcaptionsetup}{}{%
 \PassOptionsToPackage{caption=false}{subfig}}
\usepackage{subfig}
\makeatother

\makeatletter
\def\maketag@@@#1{\hbox{\m@th\normalfont\normalsize#1}}
\makeatother

\newtheorem{defn}{\protect\definitionname}
\newtheorem{prop}{\protect\propositionname}

\newtheorem{lem}{\protect\lemmaname}
\newtheorem{rem}{\protect\remarkname}
\newtheorem{defn*}{\protect\definitionname}
\newtheorem{thm}{\protect\theoremname}
\newtheorem{example}{\protect\examplename}

\providecommand{\theoremname}{Theorem}
\providecommand{\definitionname}{Definition}
\providecommand{\examplename}{Example}
\providecommand{\lemmaname}{Lemma}
\providecommand{\propositionname}{Proposition}
\providecommand{\remarkname}{Remark}
\providecommand{\corollaryname}{Corollary}

\newcommand{\Fc}{\mathcal{F}}
\newcommand{\PP}{\mathbb{P}}
\newcommand{\citet}{\cite}
\DeclareMathOperator{\nd}{\mathsf{d\kern-.15ex I}}	

\begin{document}

\title{Incorporating statistical model error into the calculation of acceptability prices of contingent claims}

\author{Martin Glanzer  \and
        Georg Ch. Pflug \and Alois Pichler
}

\date{}



\setlist[enumerate]{align=left}

\maketitle

\begin{abstract}

\noindent The determination of acceptability prices of contingent claims requires the choice of a stochastic model for the underlying asset price dynamics. Given this model, optimal bid and ask prices can be found by stochastic optimization. However, the model for the underlying asset price process is typically based on data and found by a statistical estimation procedure. We define a confidence set of possible estimated models by a nonparametric neighborhood of a baseline model. This neighborhood serves as ambiguity set for a multistage stochastic optimization problem under model uncertainty. We obtain distributionally robust solutions of the acceptability pricing problem and derive the dual problem formulation. Moreover, we prove a general large deviations result for the nested distance, which allows to relate the bid and ask prices under model ambiguity to the quality of the observed data.

\end{abstract}

\section{Introduction}
\label{sec: Introduction}

The no-arbitrage paradigm is the cornerstone of mathematical finance. The fundamental work of Harrison, Kreps and Pliska \cite{HarrisonKreps79,HarrisonPliska81,HarrisonPliska83,Kreps81} and Delbaen and Schachermayer \cite{DelbaenSchachermayer94}, to mention some of the most important contributions, paved the way for a sound theory for the pricing of contingent claims. In a general market model, the exclusion of arbitrage opportunities leads to intervals of fair prices.

Typically, the resulting no-arbitrage price bounds are too wide to provide practically meaningful information.\footnote{For example, the superreplication price for a plain vanilla call option in exponential L\'{e}vy models is given by the spot price of the underlying asset (see Cont and Tankov \cite[Prop. 10.2]{ContTankov04}), which is a trivial upper bound for the call option price.} In practice, market-makers wish to have a framework for controlling the acceptable risk when setting their spreads. Pioneering contributions to incorporate risk in the pricing procedure for contingent claims were made by Carr, Geman and Madan \cite{CarrGemanMadan01} as well as F\"ollmer and Leukert \cite{FollmerLeukert99, FollmerLeukert00}, subsequent generalizations being made, e.g., by Nakano \cite{Nakano04} or Rudloff \cite{Rudloff07}. The pricing framework of the present paper is in this spirit: by specifying acceptability functionals, an agent may control her shortfall risk in a rather intuitive manner. In particular, using the Average-Value-at-Risk ($\avar_\alpha$) will allow for a whole range of prices between the extreme cases of hedging with probability one (the traditional approach) and hedging w.r.t.\ expectation by varying the parameter $\alpha\,$.

Nowadays, there is great awareness of the epistemic uncertainty inherent in setting up a stochastic model for a given problem. For single-stage and two-stage situations, there is a plethora of available literature on different approaches to account for model ambiguity (see the lists contained in \cite[pp.\ 232--233]{PflugPichler14} or \cite[p.\ 2]{VanParysetal17}). Recently, balls w.r.t.\ the Kantorovich-Wasserstein distance around an estimated model have gained a lot of popularity (e.g., \cite{EsKu17, HanasusantoKuhn17, GaoKleywegt2016, ZhaoGuan18, Nguyenetal18, Duanetal18}), while originally proposed by Pflug and Wozabal \cite{PflugWozabal07} in 2007. However, the literature on nonparametric ambiguity sets for multistage problems is still extremely sparse. Analui and Pflug \cite{AnaluiPflug14} were the first to study balls w.r.t.\ the multistage generalization of the Kantorovich-Wasserstein distance, named nested distance,\footnote{The definition of the nested distance can be found in the Appendix.} for incorporating model uncertainty into multistage decision making. It is the aim of this article to further explore this rather uncharted territory. The classic mathematical finance problem of contingent claim pricing serves as a very well suited instance for doing so. In fact, while in the traditional pointwise hedging setup only the null sets of the stochastic model for the dynamics of the underlying asset price process influence the resulting price of a contingent claim, the full specification of the model affects the claim price when acceptability is introduced. Thus, model dependency is even stronger in the latter case, which is the topic of this paper.

Stochastic optimization offers a natural framework to deal with the problems of mathematical finance. Application of the fundamental work of Rockafellar and Wets \cite{Rock74,RockafellarWets76,RockafellarWets76a,RockafellarWets76b,RockafellarWets76c,RockafellarWets77,RockafellarWets78} on conjugate duality and stochastic programming has led to a stream of literature on those topics. King \cite{King02} originally formulated the problem of contingent claim pricing as a stochastic program. Extensions of this approach have been made, amongst others, by King, Pennanen and their coauthors \cite{KingKoivuPennanen05, PennanenKing04, KingKorf02, King02, King2010, Pennanen11, Pennanen14}, Kallio and Ziemba \cite{KallioZiemba07} or Dahl \cite{Dahl17}. The stochastic programming approach naturally allows for incorporating features and constraints of real-world markets and allows to efficiently obtain numerical results by applying the powerful toolkit of available algorithms for convex optimization problems.

The main contribution of this article is the link between statistical model error and the pricing of contingent claims, where the pricing methodology allows for a controlled hedging shortfall. The setup is inspired by practically very relevant aspects of decision making under both aleatoric and epistemic uncertainty. Given the stochastic model from which future evolutions are drawn, agents are willing to accept a certain degree of risk in their decisions. How\-ever, it may be dangerously misleading to neglect the fact that it is impossible to detect the true model without error. Thus, a distributionally robust framework, which takes the limitations of nonparametric statistical estimation into account, is required. In the statistical terminology, balls w.r.t.\ the nested distance may be seen as confidence regions: by considering all models whose nested distance to the estimated baseline model does not exceed some threshold, it is ensured that the true model is covered with a certain probability and hence the decision is robust w.r.t.\ the statistical model estimation error. In particular, we prove a large deviations theorem for the nested distance, based on which we show that a scenario tree can be constructed out of data such that it converges (in terms of the nested distance) to the true model in probability at an exponential rate. Thus, distributionally robust claim prices w.r.t.\ nested distance balls as ambiguity sets include a hedge under the true model with arbitrary high probability, depending on the available data. In other words, we provide a framework that allows for setting up bid and ask prices for a contingent claim which result from finding hedging strategies with truly calculated risks, since the important factor of model uncertainty is not neglected.

This paper is organized as follows. In Section~\ref{sec:Acceptability-pricing} we introduce our framework for acceptability pricing, i.e., we replace the traditional almost sure super-/\,subreplication requirement by the weaker constraint of an acceptable hedge. The acceptability condition is formulated w.r.t.\ one given probability model. This lowers the ask price and increases the bid price such that the bid-ask spread may be tightened or even closed. Section~\ref{sec: ambiguity models} contains the main results of this article. We weaken the assumption of one single probability model assuming that a collection of models is plausible. In particular, we define the distributionally robust acceptability pricing problem and derive the dual problem formulation under rather general assumptions on the ambiguity set. The effect of the introduction of acceptability and ambiguity into the classical pricing methodology is nicely mirrored by the dual formulations. Moreover, we give a strong statistical motivation for using nested distance balls as ambiguity sets by proving a large deviations theorem for the nested distance. Section \ref{sec:Numerical Solutions} contains illustrative examples to visualize the effect of acceptability and model ambiguity on contingent claim prices. In Section~\ref{sec: algorithmic solution} we discuss the algorithmic solution of the $\avar$-acceptability pricing problem w.r.t.\ nested distance balls as ambiguity sets. In particular, we exploit the duality results of Section~\ref{sec: ambiguity models} and the special stagewise structure of the nested distance by a sequential linear programming algorithm which yields approximate solutions to the originally semi-infinite non-convex problem. In this way, we overcome the current state-of-the-art computational methods for multistage stochastic optimization problems under non-parametric model ambiguity. Finally, we summarize our results in Section~\ref{ref: conclusio}.

\section{Acceptability pricing\label{sec:Acceptability-pricing}}

\subsection{Acceptability functionals}
\label{subsec: acceptability functionals}

The terminology introduced in this section follows the book of Pflug and R\"omisch \cite{PflugRoemisch07}. A detailed
discussion of acceptability functionals and their properties can be found therein. Intuitively speaking, an acceptability functional $\mathcal A$ maps a stochastic position $Y \in L_p(\Omega), 1<p<\infty,$ defined on a probability space $(\Omega, \mathcal{F}, \mathbb P)$,  to the real numbers extended by $-\infty$ in such a way that higher values of the position correspond to higher values of the functional, i.e., a `higher degree of acceptance'. In particular, the defining properties of an acceptability functional are \textit{translation equivariance},\footnote{$\mathcal A(Y+c) = \mathcal A(Y) + c$ for any $c \in \mathbb R$} \emph{concavity}, \emph{monotonicity},\footnote{$X \leq Y \textnormal{ a.s. } \Longrightarrow \mathcal A(X) \leq \mathcal A(Y)$} and \emph{positive homogeneity}. We assume all acceptability functionals to be \emph{version independent},\footnote{For version independent acceptability functionals, upper semi-continuity follows from concavity (see Jouini, Schachermayer and Touzi \cite{JouiniSchachermayerTouzi06}).} i.e., $\mathcal A(Y)$ depends only on the distribution of the random variable $Y$.

The following proposition is well-known. It follows directly from the Fenchel-Moreau-Rockafellar Theorem (see \cite[Th.~5]{Rock74} and \cite[Th.~2.31]{PflugRoemisch07}).

\begin{prop}
\label{thm: implications from some acceptability functional types}
An acceptability functional $\mathcal{A}$ which fulfills the above conditions has a dual representation of the form
\begin{align*}
\mathcal{A}(Y) = \inf \left\{ \mathbb{E}\left[YZ\right]\colon Z\in\mathcal{Z} \right\},
\end{align*}
where $\mathcal{Z}$ is a closed convex subset of $L^q(\Omega)$, with $1/p+1/q=1\,$. We call $\mathcal{Z}$ the superdifferential of $\mathcal{A}$. Monotonicity and translation equivariance imply that all $Z \in \mathcal{Z}$ are nonnegative densities.
\end{prop}

{\bf Assumption A1.} There exists some constant $K_1 \in \mathbb R$ such that for all $Z \in \mathcal Z$ it holds $\Vert Z \Vert_q \leq K_1\,$.

This assumption implies that $\mathcal{A}$ is Lipschitz on $L_p$:
\begin{equation}\label{LipschitzLp}
|\mathcal{A}(Y_1)-\mathcal{A}(Y_2)| \le K_1 \; \|Y_1 - Y_2\|_p \, .
\end{equation}

A good example for such an acceptability functional is the Average Value-at-Risk, $\avar_\alpha$, whose superdifferential is given by
\[\mathcal{Z} = \{Z\in L_1(\Omega)\colon 0\le Z \le 1/\alpha\text{ and } \mathbb{E}(Z) = 1 \}.\]
The extreme cases are represented by the essential infimum ($\avar_{0}(Y) := \lim_{\alpha \downarrow 0} \avar_\alpha(Y) = \operatorname{essinf}(Y)$\footnote{Strictly speaking, Assumption A1 is not respected by $\avar_0\,$. However, all our results on $\avar$--acceptability pricing will hold true also for $\avar_0\,$. In fact, this is the special case which is well treated in the literature.}) and the expectation ($\alpha=1$). Its superdifferentials are given by the set of all probability densities and just the function identically~$1$, respectively.

Other common names for the $\avar$ are Conditional-Value-at-Risk, Tail-Value-at-Risk, or Expected Shortfall. The subtleties between these terminologies are, e.g., addressed in Sarykalin et~al. \cite{Sarykalinetal08}. All our computational studies in Section \ref{sec:Numerical Solutions} and Section \ref{sec: algorithmic solution} will be based on some $\avar_\alpha$, while our theoretical results are general.

\subsection{Acceptable replications\label{sub:Hedging-to-acceptability}}

Let us now introduce the notion of acceptability in the pricing procedure for contingent claims.

As usual in mathematical finance, we consider a market model as a filtered probability space $(\Omega,\mathcal{F},\mathbb{P})$, where the filtration is given by the increasing sequence of sigma-algebras $\mathcal{F}=(\mathcal{F}_0, \mathcal{F}_1, \dots, \mathcal{F}_T)$ with $\mathcal{F}_0=\{\emptyset,\Omega\}$. The liquidly traded basic asset prices are given by a  discrete-time $\mathbb{R}_+^{m}$-valued stochastic process  $S = (S_0, \dots, S_T)$, where $S_t=(S_t^{(1)}, S_t^{(2)},\dots, S_t^{(m)})$.  We assume the filtration to be generated by the asset price process.

One asset, denoted by $S^{(1)}$, serves as num\'eraire (a risk-less bond, say).
We assume w.l.o.g. that $S_t^{(1)} =1$ a.s. If not, we may replace $(S_t^{(1)}, S_t^{(2)},\dots, S_t^{(m)})$ by $(1, S_t^{(2)}/S_t^{(1)},\dots, S_t^{(m)}/S_t^{(1)})$.

A contingent claim $C$ consists of an $\mathcal{F}$-adapted series of cash flows $C=(C_{1},\ldots,C_{T})$ measured in units of the num\'eraire. The fact that the payoff $C_{t}$ is contingent on the respective state of the market up to  time $t$ is reflected by the condition that $C$ is adapted to the filtration $\mathcal{F}$, for which we write $C \lhd \mathcal{F}$. A trading strategy $x=(x_0,\ldots, x_{T-1})$ is an $\mathcal{F}$-adapted $\mathbb{R}^{m}$-valued process with $x \lhd \mathcal{F}$.

To be more precise, let
\begin{AmSalign*}
\mathcal{L}^m_p &:= \mathbb{R}^m \times L_p^m(\Omega,\mathcal{F}_1) \times \dots \times L_p^m(\Omega,\mathcal{F}_T) \, , \\
\mathcal{L}^m_\infty &:= \mathbb{R}^m \times L_\infty^m(\Omega,\mathcal{F}_1) \times \dots \times L_\infty^m(\Omega,\mathcal{F}_{T-1}) \, , \\
\intertext{and}
\mathcal{L}^1_q &:= L_q(\Omega,\mathcal{F}_1) \times \dots \times L_q(\Omega,\mathcal{F}_T) \, .
 \end{AmSalign*}
We assume that $S \in \mathcal{L}^m_p$, $x \in \mathcal{L}_\infty^m$ and $C \in \mathcal{L}_p^1$. The norm in $L^m_p$ is given by
$$\|Y\|_p = \sum_{i=1}^m \|Y^{(i)}\|_p \, ,$$
and similarly for $L_\infty^m\,$. Notice that $x_0$ and $S_0$ are deterministic vectors.

{\bf Assumption A2.} We assume that all claims are Lipschitz-continuous functions of the underlying asset price process $S$.

\begin{defn}
\label{DefiAcceptablePricing}
Consider a contingent claim $C$ and fix acceptability functionals $\mathcal{A}_{t}$, for all $t=1,\ldots,T$. We assume that all functionals $\mathcal{A}$ have a representation given by Proposition \ref{thm: implications from some acceptability functional types}. Then the acceptable prices are given by the optimal values of the following stochastic optimization programs:
\begin{enumerate}
\item[i)]  the \textnormal{\textbf{acceptable ask price}} of $C$ is defined as
{\small
\begin{subequations}
    \begin{align}[left = (\textnormal{P})\hspace{2mm}\empheqlVert\,]
     \pi_{a}(\mathcal{A}_{1},\dots,\mathcal{A}_{T}) = \min_{x}~& x_{0}^{\top}S_{0} \notag\\
\textnormal{s.t. }& \mathcal{A}_{t}(x_{t-1}^{\top}S_{t}-x_{t}^{\top}S_{t}-C_{t})\geq0 \label{Pa}\\
&\mathcal{A}_{T}(x_{T-1}^{\top}S_{T}-C_{T})\geq0 \label{Pb}\, ,
    \end{align}
    \end{subequations}}

\item[ii)] the \textnormal{\textbf{acceptable bid price}} of $C$ is defined as
{\small
\begin{subequations}
    \begin{align}[left = (\textnormal{P}^\prime) \hspace{2mm}\empheqlVert\,]
     \pi_{b}(\mathcal{A}_{1},\dots,\mathcal{A}_{T}) = \max_x~ & x_{0}^{\top}S_{0} \notag\\
\textnormal{s.t. }  & \mathcal{A}_{t}(x_{t}^{\top}S_{t}-x_{t-1}^{\top}S_{t}+C_{t})\geq0 \label{Pprimea}\\
 & \mathcal{A}_{T}(-x_{T-1}^{\top}S_{T}+C_{T})\geq0 \, , \label{Pprimeb}
    \end{align}
    \end{subequations}
}where the optimization runs over all trading strategies $x \in \mathcal{L}_\infty^m $ for the liquidly traded assets. The constraints in \eqref{Pa} and \eqref{Pprimea} are formulated for all $t = 1,\ldots,T-1$.
\end{enumerate}
\end{defn}

To interpret Definition~\ref{DefiAcceptablePricing}, the acceptable ask price is given by the minimal initial capital required to acceptably superhedge the cash-flows $C_t$, which have to be paid out by the seller. On the other hand, the acceptable bid price corresponds to the maximal amount of money that can initially be borrowed from the market to buy the claim, such that by receiving the payments $C_t$ and always rebalancing one's portfolio in an acceptable way, one ends up with an acceptable position at maturity.

In what follows we will mainly consider the ask price problem $(\rm{P})$ and its variants. The bid price problem $(\rm{P}^\prime)$ is its mirror image and all assertions and proofs for the problem $(\rm{P})$ can be rewritten literally for problem $(\rm{P}^\prime)$.

Let $(\rm{P}^\beta)$  for $\beta=(\beta_1, \dots, \beta_T)$ be the problem $(\textnormal{P})$, where the conditions (\ref{Pa}) and (\ref{Pb}) are replaced by $\mathcal{A}_t (\cdot) \ge \beta_t$.

{\bf Assumption A3.} The optima are attained and all solutions $x$ to the problems $(\rm{P}^\beta)$, for $\beta$ in a neighborhood of 0, are uniformly bounded, i.e., $ \exists K_2 \in \mathbb R \textnormal{ s.t. } \forall x\colon \| x\|_\infty \le K_2$.

We show the following auxiliary result for the problems $(\rm{P}^\beta)$.

\begin{lem}\label{Pbetaapprox}
Let $v^\beta$ be the optimal value of $(\rm{P}^\beta)$ and $v^*$ be the optimal value of $(\rm{P})$. Then, in a neighborhood of $0$,
\begin{equation}\label{vbeta}
|v^\beta-v^*|\le 2 \bar{\beta} \cdot \|S_0\|_1
\end{equation}
where $\bar{\beta} = \sum_t |\beta_t|$.
\end{lem}

\begin{proof}
If $v^{\beta}$ is the optimal value of $(\rm{P}^\beta)$, then by inclusion of the feasible sets
\begin{eqnarray*}
v^{-|\beta|} &\le& v^* \le v^{|\beta|}\,,\\
v^{-|\beta|} &\le& v^\beta \le v^{|\beta|}\,.
\end{eqnarray*}
We have to bound $v^{|\beta|} - v^{-|\beta|}$.
Let $x_t^{*}$ be the solution of $(\rm{P}^{-|\beta|})$.  $x_t^{*}$ is not necessarily feasible for $(\rm{P}^{|\beta|})$. We modify $x_t^{*}$ in order to get feasibility for $(\rm{P}^{|\beta|})$. Let $a_t, t=1,\dots,T-1\,$, be the vector with identical components  $2 \sum_{s=t+1}^T |\beta_s| $  and let $x_t = x_t^{*}+a_t$. Then
\begin{align*}
\mathbb{E}[(x_{t-1}-x_t)^\top  S_t Z_t]&-\mathbb{E}[(x_{t-1}^{*}-x_t^{*})^\top S_t Z_t ]\\
&= \mathbb{E}[ (a_{t-1} - a_t)^\top S_t Z_t] = 2|\beta_t| \sum_{i=1}^{m} \mathbb{E}\left[S_t^{(i)} Z_t\right] \\
&\ge 2 |\beta_t| \cdot \biggl(\inf \sum_{i=1}^{m} S_t^{(i)}\biggr) \cdot \mathbb{E}[Z_t] \ge 2 |\beta_t|
\end{align*}
since $\sum_i S_t^{(i)} \ge S_t^{(1)} = 1$ and $\mathbb{E}[Z_t]=1$. By $\mathbb{E}[(x_{t-1}^{*}-x_t^{*})^\top S_t Z_t ] \ge -|\beta_t|$, one gets that
$\mathbb{E}[(x_{t-1}-x_t)^\top S_t Z_t ] \ge |\beta_t|$, i.e., $x_t$ is feasible for $(\rm{P}^{|\beta|})$. Notice that $a_0$ has all components equal to $\sum_t |\beta_t| = \bar{\beta}$. Now
$$0 \leq v^{|\beta|} - v^{-|\beta|} \le x_0^\top S_0 - x_0^{*\top} S_0 = a_0^\top S_0 = 2 \bar{\beta} \sum_i S_0^{(i)} = 2 \bar{\beta} \cdot \|S_0 \|_1,$$ which concludes the proof.
\end{proof}

Notice that the primal program $(\rm{P})$ is semi-infinite, if the constraints are written in the extensive form
$$\mathbb{E}[\left( (x_{t-1}-x_t)^\top S_t - C_t \right) Z_t ] \ge 0  \qquad \hbox{ for all } Z_t \in \mathcal{Z}_t \, ,$$
where $Z = (Z_1, \ldots, Z_T) \in \mathcal{L}^1_q$.

Lemma~\ref{lem: approx by vn} below demonstrates the validity of an approximation with only finitely many supergradients.

Since the $L_p$ spaces are separable, there exist sequences $(Z_{t,1}, Z_{t,2}, \dots)$ that are dense in $\mathcal{Z}_t$, for each $t\,$. Let
$$\mathcal{A}_{t,n}(Y) = \min \{ \mathbb{E}[Y \cdot Z_{t,i}]\colon 1 \le i \le n\}.$$
Since $Z \mapsto \mathbb E[YZ]$ is continuous in $L_p\,$, for every $Y$ in $L_p(\Omega, \mathcal{F}_t)$ it holds that 
$$\mathcal{A}_{t,n}(Y) \downarrow \mathcal{A}_t(Y),$$
as $n \to \infty$.

\begin{lem}
\label{lem: approx by vn}
Let $v^*$ be the optimal value of the basic problem $(\rm{P})$ and let $v^*_n$ be the optimal value of the similar optimization problem $(\rm{P}_n)$, where  $\mathcal{A}_t$ are replaced by $\mathcal{A}_{t,n}$. Then
 $$v_n^* \uparrow v^*.$$
\end{lem}

 \begin{proof}

Suppose the contrary, that is $\sup_n v_n^* \le  v^* - 3  \eta < v^*$ for some $\eta>0$. Introduce the notation
\begin{eqnarray*}
Y_t(x) = \left\{ \begin{array}{ll}  (x_{t-1}-x_t)^\top \, S_t - C_t & \qquad \hbox{ for } 1\le t < T\\
                        x_{T-1}^\top \, S_T - C_T            &  \qquad \hbox{ for } t=T \, .
                  \end{array} \right.
\end{eqnarray*}

By Assumption A1 and since $x \in \mathcal{L}_\infty^m$, it holds that $x \mapsto \mathcal{A}_t (Y_t(x))$ and $x \mapsto x_0^\top S_0$ are Lipschitz. Choose $0 < \delta = \eta \left[ 2 \Vert S_0 \Vert_1 K_1 (K_2+K_3+1) \right]^{-1}$ with $K_3 \ge  \|S_t\|_p$ for all $t$ .
Let $x_t^*$ be the solution of $(\rm{P})$. We may find finite sub-sigma-algebras $\tilde{\mathcal{F}}_t \subseteq \mathcal{F}_t$ such that with
 \begin{eqnarray*}
 \tilde{S}_t &=& \mathbb{E}[S_t|\tilde{\mathcal{F}}_t]   \qquad \hbox{ (componentwise)},\\
 \tilde{C}_t &=& \mathbb{E}[C_t|\tilde{\mathcal{F}}_t]\,,\\
 \tilde{x}_t^* &=& \mathbb{E}[x_t^*|\tilde{\mathcal{F}}_t] \qquad \hbox{ (componentwise)},
 \end{eqnarray*}
 we have that
 \begin{eqnarray*}
 \| S_t - \tilde{S}_t \|_p &\le& \delta, \\
 \| C_t - \tilde{C}_t \|_p &\le& \delta, \\
 \| x_t^* - \tilde{x}_t^* \|_\infty &\le& \delta.
 \end{eqnarray*}

Denote by $(\tilde{\rm{P}})$ the variant of the problem $(\rm{P})$, where the processes $(S_t)$ and $(C_t)$ are replaced by $(\tilde{S}_t)$ and $(\tilde{C}_t)$.
Similarly as before introduce the notation
\begin{eqnarray*}
\tilde{Y}_t(x) = \left\{ \begin{array}{ll}  (x_{t-1}-x_t)^\top \, \tilde{S}_t - \tilde{C}_t & \qquad \hbox{ for } 1\le t < T\\
                        x_{T-1}^\top \, \tilde{S}_T - \tilde{C}_T            &  \qquad \hbox{ for } t=T.
                  \end{array} \right.
\end{eqnarray*}
Notice that
\begin{align*}
|\mathcal{A}_t(\tilde{Y}_t(\tilde{x}^*_t))&- \mathcal{A}_t(Y_t(x^*_t))|\\
&\le  K_1 \|\tilde{Y}_t(\tilde{x}^*_t) - Y_t(x^*_t) \|_p \\
&\le K_1 [\| \tilde{x}_t^* - x_t^* \|_\infty \| \tilde{S}_t \|_p + \|x_t^*\|_\infty \|\tilde{S}_t - S_t \|_p + \|\tilde{C}_t - C_t \|_p]\\
&\le K_1 [\delta K_3 + \delta K_2 + \delta] = \eta \left[2\Vert S_0 \Vert_1\right]^{-1}.
\end{align*}

By Lemma~\ref{Pbetaapprox} we may conclude that
\begin{equation}\label{eins}
v^* \le \tilde{v}^* + \eta,
\end{equation}
where $\tilde{v}^*$ is the optimal value of $(\tilde{\rm{P}})$.  Let $(\tilde{\rm{P}}_n)$ be the variant of problem $(\tilde{\rm{P}})$, where all $\mathcal{A}_t$ are replaced by $\mathcal{A}_{t,n}$. The optimal value of $(\tilde{\rm{P}}_n)$ is denoted by $\tilde{v}_n^*$. In this finite situation we may show that  $\tilde{v}_n^* \uparrow \tilde{v}^*$.
Obviously, $\tilde{v}_{n}^*$ is a monotonically increasing sequence with $\tilde{v}_{n}^*\le \tilde{v}^*$.

It remains to demonstrate that $\lim_{n} \tilde{v}_{n}^*$ cannot be smaller than $\tilde{v}^*$. For this, let $\tilde{x}^{{n}*}$ be a solution of $(\tilde{\rm{P}}_n)$. Because of the finiteness of the filtration $\tilde{\mathcal{F}}$, the solutions of $(\tilde{\rm{P}}_n)$ as well as of $\tilde{\rm{P}}$  are just bounded vectors in some high-, but finite dimensional $\mathbb{R}^N$ and are all bounded by $K_2$.  Let $\tilde{x}^{**}$ be an accumulation point of $(\tilde{x}^{{n}*})$, i.e., we have for some subsequence that $\tilde{x}^{{n_{i}*}}\to \tilde{x}^{**}$. We show that $\tilde{x}^{**}$ satisfies the constraints of $(\tilde{\rm{P}})$.

Suppose the contrary. Then there is a $t$ such that $\mathcal{A}_t(\tilde{Y}_t(\tilde{x}^{**})) < 0$. 
This implies that there is a $Z_{t,m} \in \{ Z_{t,1}, Z_{t,2}, \dots \}$ such that
$\mathbb{E} [ \tilde{Y}_t(\tilde{x}^{**}) \cdot Z_{t,m}]<0$. However, for $n \ge m$, by construction
$\mathbb{E}[\tilde{Y}_t (\tilde{x}^{n*}) \cdot Z_{t,m})] \ge 0$ and since $\tilde{x}^{n*} \to \tilde{x}^{**}$ componentwise, then also
$\mathbb{E}[\tilde{Y}_t (\tilde{x}^{**}) \cdot Z_{t,m}] \ge 0\,.$
Since the objective function is continuous in $\tilde{x}$ this implies that $\lim_i \tilde{v}_{n_i}^*=\tilde{v}^*$ and, by monotonicity,
$\lim_{n} \tilde{v}_{n}^*=\tilde{v}^*$. We have therefore shown that  we can find an index $n$ such that
\begin{equation}\label{zwei}
\tilde{v}^* < \tilde{v}^*_n+\eta \, .
\end{equation}
Let $x^{n*}$ be the solution of $(\rm{P}_n)$ and let $\hat{x}^{n*}= \mathbb{E}[x^{n*}|\tilde{\mathcal{F}}_t]\,$. Analogously as before, one may prove that
$| \mathcal{A}_t (\tilde{Y}_t(\hat{x}^{n*}) | \le \eta \left[2\Vert S_0 \Vert_1\right]^{-1}$ and hence, by Lemma~\ref{Pbetaapprox},
\begin{equation}\label{drei}
\tilde{v}_n^* \le v_n^* + \eta.
\end{equation}
Putting (\ref{eins}), (\ref{zwei}) and (\ref{drei}) together one sees that
$$v^* \le v_n^* + 3 \eta \, ,$$
 which contradicts the assumption that $v^*_n < v^*-3 \eta$\, .
\end{proof}

We now turn to the duals of the problems $(\rm{P})$ and $(\rm{P}^\prime)$, called $(\rm{D})$ and $(\rm{D}^\prime)$, respectively. It turns out that also in our general acceptability case a martingale property appears in the dual as it is known for the case of a.s.\ super-/\,subreplication.

\begin{thm}
\label{prop: duality result for acceptability superhedging} For all $t=1,\ldots,T$, let $\mathcal{A}_{t}$ be acceptability functionals with corresponding superdifferentials $\mathcal{Z}_t$. Then, the acceptable ask price is given by

{\small
\begin{subequations}
    \begin{align}[left = (\textnormal{D})\hspace{2mm}\empheqlVert\,]
    \pi_{a}(\mathcal{A}_{1},\dots,\mathcal{A}_{T})=\sup_{\mathbb{Q}}~ & \mathbb{E}^{\mathbb{Q}}\left[\sum_{t=1}^{T} C_{t}\right] \notag\\
\textnormal{s.t. } & \mathbb{E}^{\mathbb{Q}}[S_{t+1}\vert\mathcal{F}_{t}]=S_{t}\hspace{5mm}\forall t=0,\ldots,T-1 \label{martingale constr ask}\\
 & {\left.\frac{d\mathbb{Q}}{d\mathbb{P}}\middle\vert\right.}_{\mathcal{F}_{t}}\in\mathcal{Z}_t \hspace{5mm}\forall t=1,\ldots,T \, , \label{measure change constraint ask}
    \end{align}
    \end{subequations}
}and the acceptable bid price is given by

{\small
\begin{subequations}
    \begin{align}[left = ({\textnormal{D}}^\prime)\hspace{2mm}\empheqlVert\,]
    \pi_{b}(\mathcal{A}_{1},\dots,\mathcal{A}_{T})=\inf_{\mathbb{Q}}~ & \mathbb{E}^{\mathbb{Q}}\left[\sum_{t=1}^{T}C_{t}\right] \notag\\
\textnormal{s.t. } & \mathbb{E}^{\mathbb{Q}}[S_{t+1}\vert\mathcal{F}_{t}]=S_{t}\hspace{5mm}\forall t=0,\ldots,T-1 \label{martingale constr bid}\\
 & {\left.\frac{d\mathbb{Q}}{d\mathbb{P}}\middle\vert\right.}_{\mathcal{F}_{t}}\in\mathcal{Z}_t \hspace{5mm}\forall t=1,\ldots,T \, . \label{measure change constraint bid}
    \end{align}
    \end{subequations}}
\end{thm}

\begin{proof}
The acceptable ask/\,bid price corresponds to a special case of the distributionally robust acceptable ask/\,bid price introduced in Definition~\ref{defi: distrib robust price} below, namely when the ambiguity set reduces to a singleton. Hence, the validity of Theorem~\ref{prop: duality result for acceptability superhedging} follows directly from the proof of Theorem~\ref{prop: duality result for robust accept superhedging}.
\end{proof}

\begin{rem}[Interpretation of the dual formulations]
The objective of\\the dual formulations $({\rm{D}})$ and $({\rm{D}}^\prime)$ is to maximize (minimize, resp.) the expected value of the payoffs resulting from the claim w.r.t.\ some feasible measure $\mathbb Q$. The constraints \eqref{martingale constr ask} and \eqref{martingale constr bid} require $\mathbb Q$ to be such that the underlying asset price process is a martingale w.r.t.\ $\mathbb Q$. This is well known from the traditional approach of pointwise super-/\,subreplication. The acceptability criterion enters the dual problems in terms of the constraints \eqref{measure change constraint ask} and \eqref{measure change constraint bid}, which reduce the feasible sets by a stronger condition than the two probability measures just having the same null sets. Making the feasible sets smaller obviously lowers the ask price and increases the bid price and thus gives a tighter bid-ask spread.
\end{rem}

\begin{prop}
\label{prop: claim price Lipschitz in P}
For fixed acceptability functionals $\mathcal A_1, \ldots, \mathcal A_T$, consider the acceptable ask price $\pi^{a}(\mathbb P)$ as a function of the underlying model $\mathbb P\,$. This function is Lipschitz.
\end{prop}
\begin{proof}
The assertion follows from Theorem~\ref{thm: Lipschitz continuity wrt nested dist} in the Appendix, considering the Lipschitz property of claims (Assumption A2) and the problem formulation resulting from Theorem~\ref{prop: duality result for acceptability superhedging}.
\end{proof}

\section{Model ambiguity and distributional robustness}
\label{sec: ambiguity models}

Traditional stochastic programs are based on a given and fixed probability model for the uncertainties. However, already since the pioneering paper of Scarf \cite{Scarf57} in the 1950s, it was felt that the fact that these models are based on observed data as well as the statistical error should be taken into account when making decisions. Ambiguity sets are typically either a finite collection of models or a neighborhood of a given baseline model. In what follows we study the latter case and, in particular, we use the nested distance to construct parameter-free ambiguity sets.

\subsection{Acceptability pricing under model ambiguity\label{sec:Ambiguity pricing}}

In Section \ref{sub:Hedging-to-acceptability} we defined the bid/\,ask price of a contingent claim as the maximal/\,minimal amount of capital needed in order to sub-/\,superhedge its payoff(s) w.r.t.\ to an acceptability criterion. However, the result computed with this approach heavily depends on the particular choice of the probability model. This section weakens the strong dependency on the model. More specifically, acceptable bid and ask prices shall be based on an acceptability criterion that is robust w.r.t.\ all models contained in a certain ambiguity set.

\begin{defn}
\label{defi: distrib robust price}
Consider a contingent claim $C$. Then, for acceptability functionals $\mathcal{A}_{t}$, $t=1,\ldots,T$, and an ambiguity set $\Pep$ of probability models,
\begin{enumerate}
\item[i)]  the \textnormal{\textbf{distributionally robust acceptable ask price}} of $C$ is defined as
{\fontsize{9.2}{12} \selectfont
\begin{subequations}
    \begin{align}[left = (\textnormal{PP})\empheqlVert\,]
    \pi_{a}^{\Pep}(\mathcal{A}_{1},\dots,\mathcal{A}_{T}) =\min_{x}~& x_{0}^{\top}S_{0} \notag\\
\textnormal{s.t.\ }  & \mathcal{A}_{t}^{\mathbb{P}}(x_{t-1}^{\top}S_{t}-x_{t}^{\top}S_{t}-C_{t})\geq0~~\forall\mathbb{P}\in\Pep \label{PPa}\\
 & \mathcal{A}_{T}^{\mathbb{P}}(x_{T-1}^{\top}S_{T}-C_{T})\geq0~~\forall\mathbb{P}\in\Pep \, , \label{PPb}
    \end{align}
    \end{subequations}}

\item[ii)]  the \textnormal{\textbf{distributionally robust acceptable bid price}} is defined as
{\fontsize{9.2}{12} \selectfont{
\begin{subequations}
    \begin{align}[left = (\textnormal{PP}^\prime)\empheqlVert\,]
    \pi_{b}^{\Pep}(\mathcal{A}_{1},\dots,\mathcal{A}_{T}) = \max_{x}~& x_{0}^{\top}S_{0} \notag\\
\textnormal{s.t.\ }  & \mathcal{A}_{t}^{\mathbb{P}}(x_{t}^{\top}S_{t}-x_{t-1}^{\top}S_{t}+C_{t})\geq0~~\forall\mathbb{P}\in\Pep \label{PPprimea}\\
 & \mathcal{A}_{T}^{\mathbb{P}}(-x_{T-1}^{\top}S_{T}+C_{T})\geq0~~\forall\mathbb{P}\in\Pep \, , \label{PPprimeb}
    \end{align}
    \end{subequations}}}
\end{enumerate}
where the optimization runs over all trading strategies $x \in \mathcal{L}_\infty^m $ for the liquidly traded assets. The constraints in \eqref{PPa} and \eqref{PPprimea} are formulated for all $t = 1,\ldots,T-1$ and $\mathcal{A}_t^{\mathbb{P}}$ denotes the value of the acceptability functional when the underlying probability model is given by $\mathbb{P}$.
\end{defn}

\begin{thm}
\label{prop: duality result for robust accept superhedging} Let $\Pep$ be a convex set of probability models, which is spanned by a sequence of models $(\mathbb P_1, \mathbb P_2, \ldots)\,$. Moreover, let $\Pep$ be dominated by some model $\mathbb{P}_0$ and assume all densities w.r.t.\ $\mathbb{P}_0$ to be bounded. For $t=1,\ldots,T$, let $\mathcal{A}_{t}$ be acceptability functionals with corresponding superdifferentials $\mathcal{Z}_{\mathcal{A}_{t}}$. Then, the distributionally robust acceptable ask price is given by

{\small
\begin{subequations}
    \begin{align}[left = (\textnormal{DD})\empheqlVert\,]
    \pi_{a}^{\Pep}(\mathcal{A}_{1},\dots,\mathcal{A}_{T})=\sup_{\mathbb{Q}}~ & \mathbb{E}^{\mathbb{Q}}\left[\sum_{t=1}^{T}{C}_{t}\right] \notag\\
\textnormal{s.t.}~ & \mathbb{E}^{\mathbb{Q}}\left[{S}_{t+1}\middle\vert\mathcal{F}_{t}\right]={S}_{t}\hspace{5mm}\forall t< T\\
 & \forall \, t ~ \exists \, \mathbb P \in \Pep : {\left.\frac{d\mathbb{Q}}{d\mathbb{P}}\middle\vert\right.}_{\mathcal{F}_{t}}\in \mathcal{Z}_{\mathcal{A}_{t}^{\mathbb P}} \, ,
    \end{align}
    \end{subequations}
}and the distributionally robust acceptable bid price is given by
{\small
\begin{subequations}
    \begin{align}[left = (\textnormal{DD}^\prime)\empheqlVert\,]
    \pi_{b}^{\Pep}(\mathcal{A}_{1},\dots,\mathcal{A}_{T})=\inf_{\mathbb{Q}}~ & \mathbb{E}^{\mathbb{Q}}\left[\sum_{t=1}^{T}{C}_{t}\right] \notag\\
\textnormal{s.t.}~ & \mathbb{E}^{\mathbb{Q}}\left[{S}_{t+1}\middle\vert\mathcal{F}_{t}\right]={S}_{t}\hspace{5mm}\forall t<T\\
 & \forall \,t ~ \exists \,\mathbb P \in \Pep : {\left.\frac{d\mathbb{Q}}{d\mathbb{P}}\middle\vert\right.}_{\mathcal{F}_{t}}\in \mathcal{Z}_{\mathcal{A}_{t}^{\mathbb P}} \, .
    \end{align}
    \end{subequations}}
\end{thm}

\begin{proof}
Define
\begin{align*}
\mathfrak D_t := \left \{ Z_t f_t\colon~ \exists~ \mathbb P \in \Pep \textnormal{ s.t. } Z_t \in \mathcal{Z}_{\mathcal{A}_{t}^{\mathbb P}}, {\left.\frac{d\mathbb{P}}{d\Ph}\middle\vert\right.}_{\mathcal{F}_{t}} = f_t \right\}.
\end{align*}
Then, the constraints in $(\rm{PP}^\prime)$ can be written in the form
$$\mathbb{E}^{\Ph}[(x_{t-1}-x_t)^\top S_t - C_t) \mathfrak d_t ] \ge 0 \qquad \forall \mathfrak d_t \in \mathfrak D_t \, .$$
Since all densities $f_t$ are bounded by assumption,\footnote{It would be sufficient to assume $\mathcal{Z}_{\mathcal{A}_{t}} \subseteq L_s$ and $f_t \in L_r$ such that $\frac{1}{r} + \frac{1}{s} = \frac{1}{q}$. However, for simplicity, we keep $\mathcal{Z}_{\mathcal{A}_{t}} \subseteq L_q$ and assume $f_t \in L_\infty$.} Lemma~\ref{lem: approx by vn} holds true if we replace $Z_t \in \mathcal Z_t$ by $\mathfrak d_t \in \mathfrak D_t$. It can easily be seen that for each $t$ there are sequences $(\mathfrak d_{t,1},\mathfrak d_{t,2}, \ldots)$ which are dense in $\mathfrak D_t$. Let us define

{\small
\begin{align*}
\mathfrak D_t^n := \left \{ \sum_{i=1}^{n_1} \sum_{j=1}^{n_2^{i}} \lambda_{i,j} Z_t^{j,i} f_t^{i}\colon~ \sum_{i=1}^{n_1} \sum_{j=1}^{n_2^{i}} \lambda_{i,j} = 1,  \left\vert\left\{(i,j) : 1 \leq i \leq n_1, 1\leq j \leq {n_2^{i}} \right\}\right\vert = n \right\} \, .
\end{align*}
}Then, it holds that $\mathfrak D_t^{n} \subseteq \mathfrak D_t^{n+1}$ and $\bigcup_n \mathfrak D_t^n = \mathfrak D_t$. Thus, by Lemma \ref{lem: approx by vn} we may approximate $(\rm{PP})$ by a problem of the form

{\small{
\begin{align*}[left = ({\rm{PP}}_n)\empheqlVert\,]
\min_{x}~ & x_{0}^{\top}S_{0}\\
\textnormal{s.t. }&\mathbb{E}^{\mathbb{\Ph}}\left[(-x_{t-1}^{\top}S_{t}+x_{t}^{\top}S_{t}+C_{t}) \cdot Z_{t}^{i,j}f_t^{i}\right]  \leq0\hspace{5mm}\forall t<T;\forall i \leq n_1;\forall j \leq n_2^{i}\\
&\mathbb{E}^{\mathbb{\Ph}}\left[(-x_{T-1}^{\top}S_{T}+C_{T}) \cdot Z_{T}^{i,j}f_T^{i}\right] \leq0\hspace{5mm}\forall i \leq n_1;\forall j \leq n_2^{i} \, .
\end{align*}
}}Rearranging its Lagrangian leads to the following representation of $({\rm{PP}}_n)\,$:

{\small
\begin{align}[left = \empheqlVert\,]
\inf_{x} \sup_{\lambda_{0}\geq0, \lambda_{t}^{i,j}\geq0}
~  \Bigl \{& x_{0}^{\top}\left(\lambda_{0}S_{0}-\mathbb{E}^{\Ph}\left[S_{1}W_1^n\right]\right)\nonumber \\
 & +\sum_{t=1}^{T-1}\mathbb{E}^{\Ph}\left[x_{t}^{\top}\left(S_{t}W_t^n-\mathbb{E}^{\Ph}\left[S_{t+1}W_{t+1}^n\middle\vert\mathcal{F}_{t}\right]\right)\right]\nonumber \\
 & +\sum_{t=1}^{T}\mathbb{E}^{\Ph}\left[C_{t}W_t^n\right] \Bigr\} \, ,
\label{eq: minimax distr robust}
\end{align}
}where

{\small
\begin{align*}
W_{t}^n:=\sum_{i=1}^{n_1}\sum_{j=1}^{n_2^{i}}\lambda_{t}^{i,j}Z_{t}^{i,j}f_{t}^{j}\,.
\end{align*}
}This is a finite-dimensional bilinear problem. Notice that $({\rm{PP}}_n)$ is always feasible.\footnote{This follows from the fact that a feasible solution $(x_0,\ldots,x_{T-1})$ of $({\rm{PP}}_n)$ can easily be constructed in a deterministic way, starting with $x_{T-1}\,$.} We may thus interchange the $\inf$ and the $\sup$. Carrying out explicitly the minimization in $x$, the unconstrained minimax problem~\eqref{eq: minimax distr robust} can be written as the constrained maximization problem

{\small
\begin{align*}[left = \empheqlVert\,]
\sup_{\lambda_{t}^{i,j}\geq0}~ & \sum_{t=1}^{T}\mathbb{E}^{\Ph}\left[C_{t}W_{t}^n\right]\\
\textnormal{s.t.}~ & S_{t}W_{t}^n=\mathbb{E}^{\Ph}\left[S_{t+1}W_{t+1}^n\middle\vert\mathcal{F}_{t}\right]\hspace{5mm}\forall t=1,\ldots,T\\
 & W_{t}^n = \sum_{i=1}^{n_1}\sum_{j=1}^{n_2^{i}}\lambda_{t}^{i,j}Z_{t}^{i,j}f_{t}^{j} \hspace{5mm} \forall t=1,\ldots,T\,.
\end{align*}
}Introducing a new probability measure $\mathbb Q$ defined by the Radon-Nikod\'ym derivative $\frac{d\mathbb{Q}}{d\mathbb{\Ph}} =W_T^n$, the problem can be rewritten in terms of $\mathbb Q$ in the form
{\small
\begin{align*}[left = ({\rm{DD}}_n) \empheqlVert\,]
\sup_{\mathbb{Q}}~ & \mathbb{E}^{\mathbb{Q}}\left[\sum_{t=1}^{T}{C}_{t}\right]\\
\textnormal{s.t.}~ & \mathbb{E}^{\mathbb{Q}}\left[{S}_{t+1}\middle\vert\mathcal{F}_{t}\right]={S}_{t},\hspace{5mm}\forall t=0,\ldots,T-1\\
& {\left.\frac{d\mathbb{Q}}{d\mathbb{\Ph}}\middle\vert\right.}_{\mathcal{F}_{t}} \in \mathfrak D_t^{n} \; .
\label{firstcaseclassicalacceptabilityhedgingfinalformulationwithoutinZt-1}
\end{align*}
}It is left to show that there is no duality gap in the limit, as $n\rightarrow\infty\,$. Assume that the dual problem $(\rm{DD})$ has an optimal value $\pi_{a}^{\prime}\neq\pi_{a}\,$. By the primal constraints in $(\rm{PP})$, for any dual feasible
solution $\mathbb{Q}$ it holds
$$ \mathbb{E}^{\mathbb{Q}}\left[\sum_{t=1}^{T}{C}_{t}\right] \leq\mathbb{E}^{\mathbb{P}}\left[\sum_{t=1}^{T-1}(x_{t-1}^{\top}{S}_{t}-x_{t}^{\top}{S}_{t})\cdot Z_{t} f_t+x_{T-1}^{\top}{S}_{T} \cdot Z_t f_T\right] =x_{0}^{\top}S_{0} \, .$$
Thus, the optimal primal solution $\pi_{a}$ is also greater than or equal to the optimal dual solution $\pi_{a}^{\prime}\,$. Now assume $\pi_{a}^{\prime}<\pi_{a}\,$. Then, since $\pi_{a}^{n}\uparrow\pi_{a}$ by Lemma \ref{lem: approx by vn}, there must exist some $n$ such that $\pi_{a}^{n}>\pi_{a}^{\prime}\,$. Moreover, there exists some $\mathbb{Q}^{n}$, which is dual feasible and such that $\mathbb{E}^{\mathbb{Q}^{n}}\left[\sum_{t=1}^{T}{C}_{t}\right]=\pi_{a}^{n}\,$. This is a contradiction to $\pi_{a}^{\prime}$ being the limit of the monotonically increasing sequence of optimal values of the approximate dual problems of the form $({\rm{DD}}_n)$.
Hence, $\pi_{a}^{\prime}=\pi_{a}$, i.e., it is shown that there is no duality gap in the limit.

Finally, considering the structure of $\mathfrak D_t$, the condition ${\left.\frac{d\mathbb{Q}}{d\mathbb{\Ph}}\middle\vert\right.}_{\mathcal{F}_{t}} \in \mathfrak D_t$ means that it is of the form  $Z_t f_t$, where there exists some $\mathbb P \in \Pep$ such that $Z_t \in \mathcal Z_{\mathcal A_t^{\mathbb P}}$ and ${\left.\frac{d\mathbb{P}}{d\mathbb{\Ph}}\middle\vert\right.}_{\mathcal{F}_{t}}=f_t$. This completes the derivation of the dual problem formulation $({\rm{DD}})$.
\end{proof}

\subsection{Nested distance balls as ambiguity sets: a large deviations result}
\label{subsec: nested dist balls}

In order to find appropriate nonparametric distances for probability models used in the framework of stochastic optimization, one has to observe that a minimal requirement is that it metricizes weak convergence and allows for convergence of empirical distributions. The Kantorovich-Wasserstein distance does metricize the weak topology on the family of probability measures having a first moment. Its multistage generalization, the nested distance, measures the distance between stochastic processes on filtered probability spaces. The Appendix contains the definition and interpretation of both, the Kantorovich-Wasserstein distance and the nested distance.

Realistic probability models must be based on observed data. While for single- or vector-valued random variables with finite expectation the empirical distribution based on an i.i.d. sample converges in Kantorovich-Wasserstein distance to the underlying probability measure, the situation is more involved for stochastic processes. The simple empirical distribution for stochastic processes does not converge in nested distance (cf. Pflug and Pichler \cite{PflugPichler2016}), but a smoothed version involving density estimates does.

As we show here by merging  the concepts of kernel estimations and transportation distances, one may get good estimates for confidence balls and ambiguity sets under some assumptions on regularity.

Let $\PP$ be the distribution of the stochastic process $\xi=(\xi_1, \dots, \xi_T)$ with values $\xi_t \in \mathbb{R}^m$. Notice that $\PP$ is a distribution on $\mathbb{R}^\ell$ with $\ell= m\cdot T$. Let $\PP^n$ be the probability measure of $n$ independent samples from $\PP$. If $\xi^{(j)} =(\xi_1^{(j)}, \dots, \xi_T^{(j)})$, $j=1,\dots,n$ is such a sample, then the empirical distribution $\hat{\PP}_n$ puts the weight $1/n$ on each of the paths $\xi^{(j)}$. For the construction of nested ambiguity balls, the empirical distribution has to be smoothed by convolution with a kernel function $k(x)$ for $x \in \mathbb{R}^\ell$. For a bandwidth $h>0$ to be specified later, let $k_h(x)= \frac{1}{h^\ell}k(x/h)$. 
In what follows we will work with the kernel density estimate $\hat{f}_n = \hat{\PP}_n * k_h$, where $*$ denotes convolution.

\vspace*{1em}
{\bf Assumption A4.}
\begin{enumerate}
\item The support of $\PP$ is a set $D= D_1 \times \dots \times D_T$, where $D_i$ are compact sets in $\mathbb{R}^m$;
\item $\PP$ has a Lebesgue density $f$, which is Lipschitz on $D$ with constant $L$;
\item $f$ is bounded from below and from above on $D$ by $0 < \underline{c} \le f(x) \le \overline{c}$;
\item the kernel function $k$ vanishes outside the unit ball and is Lipschitz with constant $L$;
\item the conditional probabilities $\PP_t(A \vert x) = \PP(\xi_t \in A \vert (\xi_1, \dots, \xi_{t-1}) = x)$ satisfy
\begin{equation}
\mathsf{d}\left(\PP_t\left(\cdot|x\right),\PP_t\left(\cdot|y\right)\right)\le\gamma_t\left\Vert x-y\right\Vert ,\qquad x,y\in D\label{eq:Lip}
\end{equation}
for some $\gamma_t>0$. Here, $\mathsf{d}$ denotes the Wasserstein distance for probabilities on $\mathbb{R}^m$.
\end{enumerate}

\begin{rem}
The proof of Theorem \ref{thm:LargeDeviation} below relies on the lower bound $\underline{c}$ of the density. As the denominator of the conditional density $f(x\vert y)= f(x,y)/ f(y)$ has to be estimated by density estimation as well, the bound ensures that the denominator does not vanish. In fact, the assumptions on the compact cube (point 1.) can be weakened to D being a compact set; the proof, however, is slightly more involved then. For the other technical assumptions (under point 5.) we may refer to Mirkov and Pflug \cite{Mirkov2007}.
\end{rem}

\begin{thm}[Large deviation for the nested distance]
\label{thm:LargeDeviation}
Under Assumption~A4 there exists a constant $K >0$ such that
\begin{equation}
\PP^n\left(\nd\left(\PP,\hat{\PP}_n*k_h\right)>\varepsilon\right)< \exp(- K n\varepsilon^{2\ell+4}) \,,\label{eq:7-2}
\end{equation}
for $n$ sufficiently large and appropriately chosen bandwidth $h$. Here, $\nd$ denotes the nested distance.
\end{thm}

The proof of~\eqref{eq:7-2} is based on several steps presented as propositions below. 
To start with we recall two important results for density estimates $\hat{f}_n = \hat{\PP}_n * k_h$ for densities $f$ on $\mathbb{R}^\ell$. 
\begin{prop}
\label{thm:3} Under the Lipschitz conditions for $f$ and $k$ given above, it holds that
\begin{equation}
\PP^n\left(\sup_{x\in D}\left|f(x)-\hat{f}_{n}(x)\right|>\varepsilon\right)\le \PP^n\left(\mathsf{d}\left(\hat{\PP}_{n},\PP\right)>\left(\frac{\varepsilon}{2L}\right)^{\ell+2}\right).\label{eq:7}
\end{equation}
if the bandwidth is chosen as $h=\varepsilon/(2 L)$.
\end{prop}

\begin{proof}
See Bolley et al. \cite[Prop.~3.1]{Bolleyetal07}. 
\end{proof}

\begin{prop}
\label{prop:Let}Let $f$ and $g$ be densities vanishing outside a compact set $D$ and set $\PP^{f}(A)=\int_{A}f(x)\mathrm{d}x$ resp. $\PP^{g}(A)=\int_{A}g(x)\mathrm{d}x\,$. Then their Wasserstein distance $\mathsf{d}$ is bounded by
\begin{equation}
\mathsf{d}\left(\PP^{f},\PP^{g}\right)\le2\Delta\lambda(D)\left\Vert f-g\right\Vert _{\infty}.\label{eq:2}
\end{equation}
\end{prop}
Here $\Delta$ is the diameter of $D$  and $\lambda(D)$ is the Lebesgue measure of $D$.
\begin{proof}
Cf.\ \citet[Prop.~4]{PflugPichler2016}. 
\end{proof}

The next result extends the previous for conditional densities.
\begin{prop}
\label{prop:1}Let $f$ and $g$ be bivariate densities on compact sets $\bar{D}_1 \times \bar{D}_2$ bounded by
$0<\underline{c} \le f,g \le \overline{c} <\infty$ which are sufficiently close so that $\left\Vert f-g\right\Vert _{\bar{D}_{1}\times \bar{D}_{2}}\le \underline{c}\lambda(\bar{D}_{1}\times \bar{D}_{2}) [2\Delta^{\ell}]^{-1}\,$.
Then there is a universal constant $\kappa_1$, depending on the set $\bar{D}:=\bar{D}_{1}\times \bar{D}_{2}$ only, so that the conditional densities are close as well, i.e., they satisfy
\[
\left|f(x|y)-g(x|y)\right|\le \kappa_1 \sup_{x^{\prime}\in \bar{D}_{1},y^{\prime}\in \bar{D}_{2}}\left|f(x^{\prime},y^{\prime})-g(x^{\prime},y^{\prime})\right|
\]
for all $x\in \bar{D}_{1}$ and $y\in \bar{D}_{2}$, i.e., 
\begin{equation}
\sup_{y\in \bar{D}_{2}}\left\Vert f(\cdot|y)-g(\cdot|y)\right\Vert _{\bar{D}_{1}}\le \kappa_1 \left\Vert f-g\right\Vert _{\bar{D}_{1}\times \bar{D}_{2}}.\label{eq:18}
\end{equation}
\end{prop}

\begin{proof}
To abbreviate the notation set $\varepsilon:=\sup_{x,y}\left|f(x,y)-g(x,y)\right|$ and note that $\varepsilon\le \underline{c}\lambda(\bar{D}) [2\Delta^{\ell}]^{-1}\,$. Consider
the marginal density $f(y):=\int_{\bar{D}_{1}}f(x,y)\mathrm{d}x$ ($g(y):=\int_{\bar{D}_{1}}g(x,y)\mathrm{d}x$,
resp.). It holds that
\[
\left|f(y)-g(y)\right|\le\int_{\bar{D}_{1}}\left|f(x,y)-g(x,y)\right|\mathrm{d}x\le\int_{\bar{D}_{1}}\varepsilon \, \mathrm{d}x\le\Delta^{\ell }\cdot\varepsilon\,.
\]
Clearly $|f(y)|\ge \underline{c}\lambda(\bar{D}_{1})$,
where $\lambda(\bar{D}_{1})$ is the Lebesgue measure of $\bar{D}_{1}$ and therefore
\begin{equation}
\left|\frac{f(y)-g(y)}{f(y)}\right|\le\frac{\Delta^{\ell}}{\underline{c}\lambda(\bar{D}_{1})}\cdot\varepsilon\le\frac{1}{2}\,.\label{eq:5}
\end{equation}

The elementary inequality $\frac{1}{1+x}\le1+2\left|x\right|$ is
valid for $x\ge-\nicefrac{1}{2}$. With~\eqref{eq:5} it follows
that
\begin{align*}
g(x|y)-f(x|y) & =\frac{g(x,y)}{g(y)}-\frac{f(x,y)}{f(y)}=\frac{g(x,y)}{f(y)}\cdot\frac{1}{1+\frac{g(y)-f(y)}{f(y)}}-\frac{f(x,y)}{f(y)}\\
 & \le\frac{g(x,y)}{f(y)}\left(1+2\frac{|g(y)-f(y)|}{f(y)}\right)-\frac{f(x,y)}{f(y)}\\
 & =\frac{g(x,y)-f(x,y)}{f(y)}+2\frac{g(x,y)}{f(y)}\frac{|g(y)-f(y)|}{f(y)}\\
 & \le\frac{\varepsilon}{\underline{c}\lambda(\bar{D}_{1})}+2\frac{\overline{c}}{\underline{c} \lambda(\bar{D}_1)}\frac{\Delta^{\ell }}{\underline{c}\lambda(\bar{D}_{1})}\cdot\varepsilon\le \kappa_1 \varepsilon
\end{align*}
with $\kappa_1=\frac{1}{\underline{c}\lambda(\bar{D}_{1})}+\frac{2\overline{c}\Delta^{\ell}}{(\underline{c}\lambda(\bar{D}_{1}))^2}$.
The assertion of the proposition finally follows by exchanging the
roles of the densities $f$ and $g$. 
\end{proof}

\begin{thm}
Given Assumption A4 there exists a constant
$\kappa_2$ such that
\begin{equation}
\PP^n \left(\sup_{y\in \bar{D}_{2}}\mathsf{d}\left(\PP^{f(\cdot|y)},\PP^{\hat{f}_{n}(\cdot|y)}\right)>\varepsilon\right)\le \exp(- \kappa_2  n\varepsilon^{2\ell+4})\label{eq:5-1}
\end{equation}
 for all $\varepsilon>0$ and $n$ sufficiently large.
\end{thm}

\begin{proof}
It follows from~\eqref{eq:2} and~\eqref{eq:18} that
\[
\mathsf{d}\left(\PP^{f(\cdot|y)},\PP^{\hat{f}_{n}(\cdot|y)}\right)\le \kappa_3 \left\Vert f(\cdot|y)-\hat{f}_{n}(\cdot|y)\right\Vert _{\infty}\le \kappa_3 \left\Vert f-\hat{f}_{n}\right\Vert _{\infty}
\]
for $\kappa_3=2 \Delta \lambda(D) \kappa_1$. Recall the large deviation result from \citet[Th.~2.8]{Bolleyetal07}, which is given by
$$\PP^n(\mathsf{d}(\hat{\PP}_n,\PP) > \eta) \le \exp(-n \kappa^\prime \eta^2)\, ,$$
for some universal constant $\kappa^\prime$ depending on the Lipschitz constants of $f$ and $k$ only.

With~\eqref{eq:7} it follows that
\begin{align*}
\PP&\left(\sup_{y\in \bar{D}_{2}}\mathsf{d}\left(\PP^{f(\cdot|y)},\PP^{\hat{f}_{n}(\cdot|y)}\right)>\varepsilon\right)\le \PP\left(\left\Vert f-\hat{f}_{n}\right\Vert _{\infty}>\frac{\varepsilon}{\kappa_3}\right)\\
&\le \PP^n\left(\mathsf{d} \left(\hat{\PP}_{n},\PP\right) > \frac{\varepsilon^{\ell+2}}{(2L\kappa_3)^{\ell+2}}\right) \le \exp\left\{-\kappa^\prime n\left(\frac{\varepsilon^{\ell+2}}{(2L\kappa_3)^{\ell+2}}\right)^{2}\right\} \, .
\end{align*}
Setting $\kappa_2:=\kappa^\prime (2L\kappa_3)^{-2\ell-4}$
in~\eqref{eq:5-1} reveals the result.
\end{proof}
\begin{proof}[Theorem~\ref{thm:LargeDeviation}]
The previous theorem will be applied to the conditional densities of $\xi_t$ given the past $\xi_1, \dots, \xi_{t-1}$. Thus the sets $\bar{D}_i$ are interpreted as
$\bar{D}_1 = D_t$ and $\bar{D}_2 = D_1 \times \dots \times D_{t-1}$. For the probability measure $\PP$ satisfying~\eqref{eq:Lip} and
any other measure $\tilde{\PP}$ satisfying $\mathsf{d}\left(\PP_t\left(\cdot|x\right),\tilde{\PP}_t\left(\cdot|x\right)\right)\le\varepsilon_{t}$
at stage $t$ we have that
\[
\nd\left(\PP,\tilde{\PP}\right)\le\sum_{t=1}^{T}\varepsilon_{t}\gamma_{t}\prod_{s=t+1}^{T}(1+\gamma_{s}),
\]
see \citet[Sec.~4.2]{PflugPichler14} or \citet{Mirkov2007}.

We employ the results elaborated above for $\tilde{\PP}:=\hat{\PP}_{n}*k_{h}$.
Then
\begin{align*}
&\PP^n\left(\nd\left(\PP,\hat{\PP}_n*k_{h}\right)>\varepsilon\right) \\
\le ~&\PP^n\left(\sum_{t=1}^{T} \mathsf{d} \left(\PP_{t}\left(\cdot|x_{t}\right),\tilde{\PP}_{t}\left(\cdot|x_{t}\right)\right)\gamma_{t}\prod_{s=t+1}^{T}(1+\gamma_{s})>\varepsilon\right)\\
=~&\sum_{t=1}^{T}\PP^n \left(\mathsf{d}\left(\PP_{t}\left(\cdot|x_{t}\right),\tilde{\PP}_{t}\left(\cdot|x_{t}\right)\right)>\frac{\varepsilon}{T\gamma_{t}\prod_{s=t+1}^{T}(1+\gamma_{s})}\right).
\end{align*}
We employ~\eqref{eq:5-1} to deduce that
\[
\PP^n \left(\nd\left(\PP,\hat{\PP}_n*k_{h}\right)>\varepsilon\right)\le\sum_{t=1}^{T}e^{-\kappa_2  n\varepsilon_{t}^{2\ell+4}}
\]
with $\varepsilon_{t}:=\varepsilon [T\gamma_{t}\prod_{s=t+1}^{T}(1+\gamma_{s})]^{-1}$.

The desired large deviation result follows for $n$ sufficiently large
for any $K<\min_{t\in\left\{ 1,\dots,T\right\} } \kappa_2 \left[\left(T\gamma_{t}\prod_{s=t+1}^{T}(1+\gamma_{s})\right)^{2\ell+4}\right]^{-1}$. 
\end{proof}

The smoothed model $\hat{\mathbb{P}}_n*k_{h}$ is not yet a tree, but by Theorem~\ref{thm: approx by tree} of the Appendix one may find\footnote{See \cite[Chap.~4]{PflugPichler14} for methods to efficiently construct multistage models/\,scenario trees from data.} a finite tree process $\bar{\mathbb{P}}_n$, which is arbitrarily close to it. Therefore, by eventually increasing the probability bound in \eqref{eq:7-2} by another constant factor, it holds true also for $\bar{\mathbb{P}}_n\,$.

\begin{rem}
From a statistical perspective, the results contained in this section represent a strong motivation to use nested distance balls as ambiguity sets for general stochastic optimization problems on scenario trees constructed from observed data. In particular, the distributionally robust acceptable ask price allows the seller of a claim to invest in a trading strategy which gives an acceptable superhedge of the payments to be made under the \emph{true} model with arbitrary high probability, given sufficient available data.
\end{rem}

\section{Illustrative examples \label{sec:Numerical Solutions}}

One may summarize the results of the previous sections in the following way: If the martingale measure is not unique (`incomplete market'), then typically there is a positive bid-ask spread in the (pointwise) replication model. This spread does also exist in the acceptability model. However, if the acceptability functional is the $\avar_{\alpha}$, then by changing $\alpha$ we can get the complete range between the replication model ($\alpha\rightarrow0)$ and the expectation model ($\alpha=1)$. At least in the latter case, but possibly even for some $\alpha < 1\,$, there is no bid-ask spread and thus a unique price. On the other hand, model ambiguity widens the bid-ask spread: The more models are considered, i.e., the larger the radius of the ambiguity set, the wider is the bid-ask spread. For illustrative purposes, let us look at the simplest form of examples which demonstrate these effects.

\begin{example}
\label{ex: ternary tree 1}Consider a three-stage ternary tree, where the paths are uniformly distributed and given by the columns of the matrix

{\small\[\begin{bmatrix} 100& 100& 100& 100& 100& 100& 100& 100& 100 \\ 110& 110& 110& 100& 100& 100& 90& 90& 90 \\ 112& 110& 108& 102& 100& 98& 92& 90& 88 \end{bmatrix} \, .\]}

Since infinitely many equivalent martingale measures can be constructed on this tree, there is a considerable bid-ask spread for the pointwise replication model, which corresponds to the $\avar_{\alpha}$-acceptability pricing model with $\alpha=0$. However, by increasing $\alpha$ for both contract sides, the bid-ask spread gets monotonically smaller. For $\alpha=1$, there is no bid-ask spread, since all martingale measures coincide in their expectation and both buyer and seller only consider expectation in their valuation. Figure \ref{fig: AcceptabilityBAspread} visualizes this behavior for the price of a call option struck at $95\%$: the bid price increases with $\alpha$, while the ask price decreases. For $\alpha=1$ they coincide.

Computationally, $\avar$--acceptability pricing on scenario trees boils down to solving a linear program (LP). It is thus straightforward to implement and the problem scales with the complexity of LPs.
\end{example}

\begin{example}
In contrast, one may consider a three-stage binary tree model with uniformly distributed scenarios given by the columns of the matrix

{\small\[\begin{bmatrix}100& 100& 100& 100 \\ 105& 105& 95& 95 \\ 108& 102& 98& 92 \end{bmatrix} \, .\]}

This tree can carry only one single martingale measure. In such a model, the change of acceptability levels does not change the price, since also under weakened acceptability the price is determined by a martingale measure, namely the unique one (in case $\alpha$ is small enough such that it is feasible). However, in an ambiguity situation, a bid-ask spread may appear, since there are typically many martingale measures contained in ambiguity sets.
We consider nested distance balls around the baseline tree, where we keep the uniform distribution of the scenarios for simplicity, but allow the values of the process to change.\footnote{This is a non-convex problem. The results in Figure~\ref{fig: AmbiguityBAspread} are based on the standard nonlinear solver of a commercial software package (MATLAB 8.5 (R2015a), The MathWorks Inc., Natick, MA, 2015.), which finds (local) optima for our small instance of a problem.} The result for a call option struck at $95\%$ can be seen in Figure~\ref{fig: AmbiguityBAspread}. While there is a unique price for small radii $\varepsilon$ of the nested distance ball, an increasing bid-ask spread appears for larger values of $\varepsilon$.
\end{example}

\begin{figure}
\subfloat[\label{fig: AcceptabilityBAspread} Acceptability: The bid-ask spread tightens for increasing acceptability.]{\includegraphics[scale=0.29]{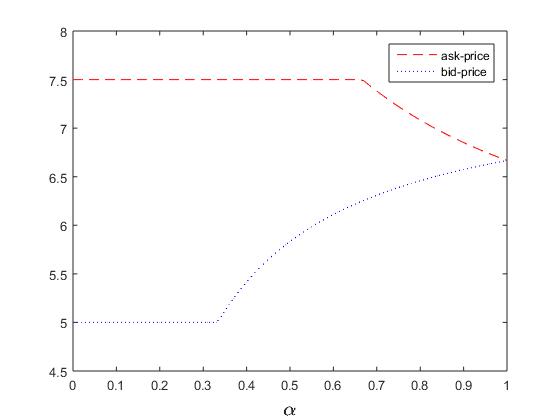}

}\hfill{}\subfloat[\label{fig: AmbiguityBAspread}Ambiguity: A bid-ask spread opens for increasing ambiguity.]{\includegraphics[scale=0.29]{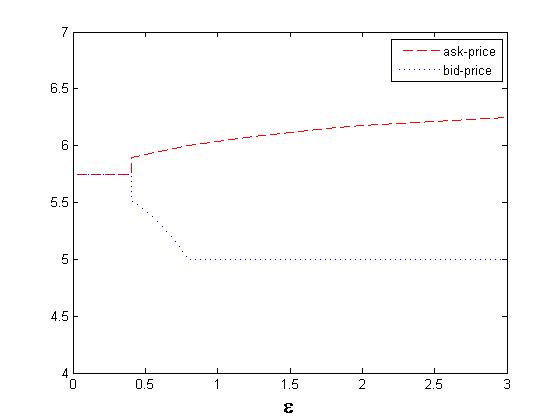}}\caption{Distributionally robust acceptability pricing: The bid-ask spread as a function of the acceptability level $\alpha$ and the ambiguity radius $\varepsilon\,$.}
\end{figure}

\section{Algorithmic solution}
\label{sec: algorithmic solution}

The nested distance between two given scenario trees can be obtained by solving an LP. However, the distributionally robust $\avar$--acceptability pricing problem w.r.t.\ nested distance balls as ambiguity sets results in a highly non-linear, in general non-convex problem. Therefore, we assume the tree structure to be given by the baseline model. In particular, it is assumed that different probability models within the ambiguity set differ only in terms of the transition probabilities; state values and the information structure are kept fixed.

Still, distributionally robust acceptability pricing is a semi-infinite non-convex problem. The only algorithmic approach available in the literature for similar problems is based on the idea of successive programming (cf.\ \cite[Chap.~7.3.3]{PflugPichler14}): an approximate solution is computed by starting with the baseline model only and alternately adding worst case models and finding optimal solutions. However, for typical instances of tree models this is computationally hard, as it involves the solution of a non-convex problem in each iteration step.

Hence, we tackle the dual formulation presented in Theorem~\ref{prop: duality result for robust accept superhedging}. The structure of the nested distance enables an iterative approach. Algorithm~\ref{algo: dual} finds an approximate solution by solving a sequence of linear programs. Based on duality considerations and algorithmic exploitation of the specific stagewise transportation structure inherent to the nested distance, the algorithm approximates the solution of a semi-infinite non-convex problem by a sequence of LPs. The current state-of-the-art method, on the other hand, requires the solution of a non-convex program in each iteration step. Clearly, a sequential linear programming approach improves the performance considerably.\footnote{For our implementations, the speed-up factor for a test problem was on average about 100. However, this may depend heavily on the implementation and the problem.} Moreover, our algorithm turned out to find feasible solutions in many cases where our implementation of a successive programming method fails to do so.

Let us extend the concept of the nested distance to subtrees, iteratively from the leaves to the root ('top-down'). For two scenario trees (here with identical filtration structures), define $\nd_T(i,j)$ as the distance of the paths leading to the leave nodes $i,j \in \mathcal N_T$. Moreover, define \[\nd_t(k,l) := \sum_{i \in k+} \sum_{j \in l+} \pi(i,j \vert k,l)\nd_{t+1}(i,j)\,,\] for all nodes $k,l \in \mathcal N_t$, where $0 \leq t < T\,$. Then, the nested distance between the two trees is given by $\nd_0(1,1)\,$. This stagewise backwards approach (cf. \cite[Alg.~2.1]{PflugPichler14}) is the basic idea of Algorithm~\ref{algo: dual}. As we assume the tree structure to be fixed, Algorithm~\ref{algo: dual} iterates through the tree in the same top-down manner and searches for the optimal solution in each stage, while ensuring that the nested distance constraint remains satisfied. The variables are the conditional transition probabilities under $\mathbb Q$, i.e., $q_i := \mathbb Q[i \vert i-]$, as well as the transportation subplans $\pi(i,j \vert i-,j-)$, as defined in the Appendix. We use the notation $n-$ for the immediate predecessor of some node $n$. As the measure $\mathbb P$ is in fact not needed explicitly since it is given by the transportation plan from $\hat{\mathbb P}\,$, condition (4.3) in Algorithm \ref{algo: dual} serves to ensure that it is still well-defined implicitly (note that always some node $\tilde{k} \in \mathcal N_{t-1}$ needs to be fixed). Condition (1) ensures that $\mathbb Q$ is a martingale measure, $\mathbb Q$ represents conditional probabilities by condition (2), condition (3) corresponds to the constraint on the measure change ($d\mathbb Q / d\mathbb P \leq 1 / \alpha$) resulting from the primal $\avar_\alpha$--acceptability conditions, and (4.1) -- (4.3) represent the constraint that there must be one $\mathbb P$ contained in the nested distance ball such that condition (3) holds.

The algorithm optimizes the variables stagewise top-down. The optimal solution at stage $t+1$ depends on the values of the variables for all stages up to stage $t$, which result from the previous iteration step. Therefore, the algorithm iterates as long as there is further improvement possible at some stage, given updated variable values for the earlier stages of the tree. Otherwise, it terminates and the optimal solution of our approximate problem is found.

\begin{breakablealgorithm}
\caption{ Acceptability pricing under model ambiguity. }
\label{algo: dual}
{\small
\begin{algorithmic}[1]
\noindent Start with some feasible model $\mathbb P$ in the nested distance ball around $\hat{\mathbb P}$. Initialize $\pi_{\textnormal{old}}$ by assigning the optimal transportation plan between $\mathbb P$ and $\hat{\mathbb P}$ and initialize 'oldprice'.

\algblock[Iteration]{Iteration}{EndIteration}
\Iteration
\State [newprice, $\pi_{\textnormal{new}}$] $\gets \textproc{\ref{func: getprice}}(\pi_{\textnormal{old}})$
\If {(oldprice == newprice)}
\State \Return oldprice
\Else
\State oldprice $\gets$ newprice, $\pi_{\textnormal{old}} \gets \pi_{\textnormal{new}}$
\State \textbf{Iterate}
\EndIf
\EndIteration

\Statex

\Function{GetPrice}{$\tilde{\pi}$}
\funclabel{func: getprice}
\For {$t$ \textbf{from} $T$ \textbf{to} $1$} \textbf{solve}
{\scriptsize
\begin{align*}[left = \empheqlVert\,]
\max_{\{ q_i, ~\pi(i,j \vert k,l) ~:~  i,j \in \mathcal N_t \}} ~& \mathbb E^{\mathbb Q} \left[ \sum_{\tau=t}^{T} C_\tau \middle \vert \mathcal F_{t-1} \right] \\
\textnormal{s.t.} ~&\\
& (1)~ \sum_{i \in k+} q_i \cdot x_i = x_k \hspace{1cm} \forall k \in \mathcal N_{t-1} \\
& (2)~ \sum_{i \in k+} q_i = 1 \hspace{1cm} \forall k \in \mathcal N_{t-1} \\
& (3)~ -\sum_{i \in k+} \pi(i,j \vert k,l) + \alpha \cdot q_j \leq 0 \hspace{1cm} \forall j \in \mathcal N_{t} \\
& (4.1)~ \sum_{i \in \mathcal N_t} \sum_{j \in \mathcal N_t} \pi(i,j \vert i-,j-) \cdot \tilde{\pi}(i-,j-) \cdot \nd_t(i,j) \leq \varepsilon \\
& (4.2)~ \sum_{j \in l+} \pi(i,j \vert k,l) =  \hat{\mathbb P}[i \vert k] \hspace{1cm} \forall l \in \mathcal N_{t-1}, \forall i \in \mathcal N_{t}\\
& (4.3)~ \sum_{i \in k+} \pi(i,j \vert k,l) = \sum_{i \in \tilde{k}+} \pi(i,j \vert \tilde{k},l) \hspace{3mm} \forall k \in \mathcal N_{t-1}, \forall j \in \mathcal N_{t} \\
&(5) \hspace {5mm} q_i, \pi(i,j \vert i-,j-) \in [0,1] \hspace{1cm} \forall i,j \in \mathcal N_t
\end{align*}}
\EndFor
\State price $\gets \mathbb E^{\mathbb Q} [ \sum_{t=1}^{T} C_t]$, construct transportation plan $\pi(\cdot, \cdot )$ from subplans $\pi(\cdot, \cdot \vert \cdot, \cdot)$
\State \Return [price, $\pi$]
\EndFunction
\end{algorithmic}}
\end{breakablealgorithm}

\begin{example}
Consider the price of a plain vanilla call option struck at 95, in the Black-Scholes model with parameters $S_0 = 100, r = 0.01, \sigma = 0.2, T = 1$. Applying optimal quantization techniques (see, e.g., \cite[Chap.~4]{PflugPichler14} for an overview) to discretize the lognormal distribution, we construct a scenario tree with 500 nodes. While there exists a unique martingale measure (and thus a unique option price) in the Black-Scholes model, the discrete approximation allows for several martingale measures (and thus a positive bid-ask spread). Figure~\ref{BAspreadBSmodel} visualizes the bid-ask spread as a function of the $\avar$--acceptability level $\alpha$ and the radius $\varepsilon$ of the nested distance ball used as model ambiguity set. For $\alpha \rightarrow 1$ and $\varepsilon = 0$, the spread closes and the resulting price approximates the true Black-Scholes price up to 4 digits. For illustrative purposes, the spread between the bid and the ask price surface is shown from two perspectives.

\begin{figure}
\begin{minipage}{0.45\textwidth}
\begin{center}
\includegraphics[width=\textwidth]{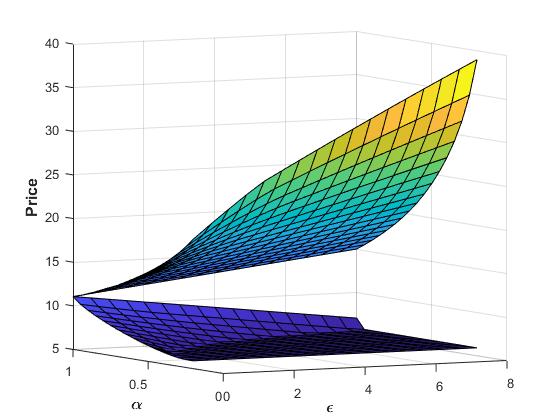}
\end{center}
\end{minipage}
\hfill
\begin{minipage}{0.45\textwidth}
\begin{center}
\includegraphics[width = \textwidth]{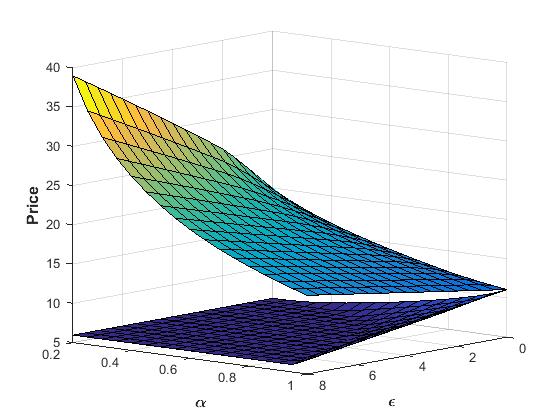}
\end{center}
\end{minipage}
\caption{The bid-ask spread as a function of acceptability and ambiguity.}
\label{BAspreadBSmodel}
\end{figure}
\end{example}

\section{Conclusion}
\label{ref: conclusio}

In this paper we extended the usual methods for contingent claim pricing into two directions. First, we replaced the replication constraint by a more realistic acceptability constraint. By doing so, the claim price does explicitly depend on the stochastic model for the price dynamics of the underlying (and not just on its null sets). If the model is based on observed data, then the calculation of the claim price can be seen as a statistical estimate. Therefore, as a second extension, we introduced model ambiguity into the acceptability pricing framework and we derived the dual problem formulations in the extended setting. Moreover, we used the nested distance for stochastic processes to define a confidence set for the underlying price model. In this way, we link acceptability prices of a claim to the quality of observed data. In particular, the size of the confidence region decreases with the sample size, i.e., the number of observed independent paths of the stochastic process of the underlying. For a given sample of observations, the ambiguity radius indicates how much the baseline ask/\,bid price should be corrected to safeguard the seller/\,buyer of a claim against the inherent statistical model risk, as Section~\ref{sec: algorithmic solution} illustrates.

\bibliographystyle{plain}
\bibliography{References1}

\section*{Appendix}

{\bf Distances for random variables and stochastic processes.} Recall the definition of the Kantorovich-Wasserstein distance $\mathsf{d}(P,\tilde{P})$ for two (Borel) random distributions $P$ and $\tilde{P}$ on $\mathbb{R}^m$:
{\small
 \begin{align*}[left = \empheqlVert\,]
\mathsf{d}(P,\tilde{P}) :=\inf_{\pi} & \iint \| \omega - \tilde{\omega}\|~ \pi(d\omega,d\tilde{\omega})\\
\textnormal{s.t. } & \pi\left(A\times \mathbb{R}^m \right)= P(A) \\
 & \pi\left(\mathbb{R}^m \times B \right)=\tilde{P}(B) \, .
\end{align*}
}Here, $\pi$ runs over all Borel measures on $\mathbb{R}^m \times \mathbb{R}^m$ with given marginals $P$ resp. $\tilde{P}$.
These measures are called {\em transportation plans}.  If $\xi$ and $\tilde{\xi}$ are $\mathbb{R}^m$-valued random variables, then their distance is defined as the distance of the corresponding image measures $P^\xi$ resp. $P^{\tilde{\xi}}$.

Pflug and Pichler \cite{Pflug09, PflugPichler12} introduced the notion of the nested distance as a generalization of the Kantorovich-Wasserstein distance for $\mathbb{R}^m$-valued stochastic processes $\xi=(\xi_1, \dots, \xi_T)$ and its image measures $\mathbb{P}$ on $\mathbb{R}^{mT}$. Let $\Fc=(\Fc_1, \dots, \Fc_T)$ be the filtration composed of the sigma-algebras $\Fc_t$ generated by the component projections $(\xi_1, \dots, \xi_T) \mapsto (\xi_1, \dots, \xi_t)\,$. Moreover, let for $\xi = (\xi_1, \dots, \xi_T) \in \mathbb{R}^{mT}$ the distance be defined as $
\| \xi - \tilde{\xi}\| := \sum_{t=1}^T \|\xi_t - \tilde{\xi}_t \|$.

\begin{defn}
The nested distance $\nd$ for distributions $\mathbb{P}$ and $\tilde{\mathbb{P}}$ is defined as
{\small
\begin{align*}[left = \empheqlVert\,]
\nd(\mathbb{P},\tilde{\mathbb{P}}):=\inf_{\pi} & \iint \| \xi(\omega) - \tilde{\xi}(\tilde{\omega})\| ~\pi(d\omega,d\tilde{\omega})\\
\textnormal{s.t. } & \pi\left(A\times\mathbb{R}^m \middle|\mathcal{F}_{t}\otimes\tilde{\mathcal{F}_{t}}\right)=\left.\mathbb{P}\Bigl[A\right\vert \mathcal{F}_{t}\Bigr]\hspace{5mm}A\in\mathcal{F}_{T};~t=1,\ldots,T\\
 & \pi\left(\mathbb{R}^m \times B\middle|\mathcal{F}_{t}\otimes\tilde{\mathcal{F}_{t}}\right)=\tilde{\mathbb{P}}\left[B\middle|\tilde{\mathcal{F}_{t}}\right]\hspace{5mm}B\in\tilde{\mathcal{F}}_{T};~t=1,\ldots,T.
\end{align*}}
\end{defn}
To interpret this definition, the nested distance between two multistage probability distributions is obtained by minimizing over all transportation plans $\pi$, which are compatible with the filtration structures. For a single period (i.e., $T=1$), the nested distance coincides with the Kantorovich-Wasserstein distance. The following basic theorem for stability of multistage stochastic optimization problems was proved by Pflug and Pichler \cite[Th.~6.1]{PflugPichler12}.

\begin{thm}
\label{thm: Lipschitz continuity wrt nested dist}
Let $\mathbb{P}$ and $\tilde{\mathbb{P}}$ be nested distributions with filtrations $\Fc$ and $\tilde{\Fc}$, respectively. Consider the multistage stochastic optimization problem
\[
v(\mathbb{P}):=\inf\left\{ \mathbb{E}^{\mathbb{P}} [Q(\xi,x)]\colon x\in\mathbb X, x\lhd\Fc \right\},
\]
where $Q$ is convex in the decisions $x=(x_1,\dots, x_T)$ for any $\xi$ fixed, and Lipschitz with
constant $L$ in the scenario process $\xi=(\xi_1, \dots, \xi_T)$ for any $x$ fixed. The set $\mathbb X$ is assumed to be convex and the constraint $x\lhd \Fc$ means that the decisions can be random variables, but must be adapted to the filtration $\Fc$, i.e., must be nonanticipative. Then the objective values $v(\mathbb{P})$ and $v(\tilde{\mathbb{P}})$ satisfy
\[
\left|v(\mathbb{P})-v(\tilde{\mathbb{P}})\right|\le L\cdot\nd(\mathbb{P},\tilde{\mathbb{P}}) \, .
\]
\end{thm}

Finite scenario trees are much easier to work with than general stochastic processes. For finite trees, where every node $m$ has a unique predecessor, we write $m+$ for the set of its immediate successors. Denote by $\mathcal{N}_{t}$ the set of all nodes at stage $t$ of the tree model $\mathbb{P}$. For a node $i \in m+$ let $\mathbb{P}[i|m]$ be the conditional transition probability from $m$ to $i\,$.

\begin{defn}
\label{def: nest dist on tree}
The nested distance for scenario trees $\mathbb{P}$ and $\mathbb{\tilde{P}}$ is defined as
{\small
\begin{align}[left = \empheqlVert\,]
\begin{split}\nd(\mathbb{P},\tilde{\mathbb{P}}) := \min_{\pi \geq 0}\  & \sum_{i}\sum_{j}\pi_{i,j}\cdot D_{i,j}\\
\textnormal{s.t. } & \sum_{j \in l+}\pi(i,j|k,l)=\mathbb{P}[i|k]\hspace{5mm}\forall i \in k+;\forall(k,l)\in(\mathcal{N}_{t}\times\tilde{\mathcal{N}}_{t}); 1 \leq t < T\\
 & \sum_{i \in k+}\pi(i,j|k,l)=\tilde{\mathbb{P}}[j|l]\hspace{5mm}\forall j \in l+;\forall(k,l)\in(\mathcal{N}_{t}\times\tilde{\mathcal{N}}_{t}); 1 \leq t < T\\
 & \pi_{i,j} \geq 0 \textnormal{ and } \sum_{i}\sum_{j}\pi_{i,j}=1 \, .
\end{split}
\label{eq: definestdisttree}
\end{align}
}The matrix $\pi$ of transportation plans and the matrix $D$ carrying the pairwise distances of the paths are defined on $\mathcal{N}_{T}\times\tilde{\mathcal{N}}_{T}$. The conditional joint probabilities $\pi(i,j|k,l)$ in \eqref{eq: definestdisttree} are given by $ \pi(i,j|k,l)=\pi_{i,j} \cdot [\sum\limits _{i^{\prime}\in k+}\sum\limits _{j^{\prime}\in l+}\pi_{i^{\prime},j^{\prime}}]^{-1} \, .$
\end{defn}

\noindent{\bf Approximation of random processes by finite trees.} The subsequent result follows from \cite[Prop.~4.26]{PflugPichler14}.

\begin{thm}
\label{thm: approx by tree}
If the stochastic process $\xi=(\xi_1, \dots, \xi_T)$ satisfies the Lipschitz condition given in Assumption A4.5 in Section \ref{subsec: nested dist balls}, then for every $\varepsilon > 0$ there is a stochastic process with distribution $\tilde{\mathbb{P}}$, which is  defined on a finite tree and which satisfies
$$\nd(\mathbb{P}, \tilde{\mathbb{P}}) \le \varepsilon,$$
where $\mathbb P$ is the distribution of $\xi$ on the filtered space $(\Omega, \mathcal F)$.
\end{thm}

\end{document}